\documentclass[draftclsnofoot]{IEEEtran}
\onecolumn
\usepackage{latexsym}
\usepackage{cite}
\usepackage{amssymb}
\usepackage{bm}
\usepackage{amsmath, amssymb}
\usepackage[dvips]{graphics}
\usepackage{graphicx}
\usepackage{psfrag}

\newtheorem{theorem}{Theorem}
\newtheorem{lemma}{Lemma}

\newtheorem{corollary}{Corollary}

\newtheorem{example}{Example}

\begin{document}
\title{Interpretations of Directed Information in Portfolio Theory,
  Data Compression, and Hypothesis Testing}
\author{Haim H.~Permuter,  Young-Han Kim, and Tsachy Weissman
\thanks{The work was supported by the NSF grant CCF-0729195
and BSF grant 2008402. Author's emails: haimp@bgu.ac.il,
yhk@ucsd.edu, and tsachy@stanford.edu.} }

\maketitle

\begin{abstract}
We investigate the role of Massey's directed information in portfolio
theory, data compression, and statistics with causality
constraints. In particular, we show that directed information is an
upper bound on the increment in growth rates of optimal portfolios in
a stock market due to {causal} side information.  This upper bound is
tight for gambling in a horse race, which is an extreme case of stock
markets. Directed information also characterizes the value of {causal}
side information in instantaneous compression and quantifies the
benefit of {causal} inference in joint compression of two stochastic
processes. In hypothesis testing, directed information evaluates the
best error exponent for testing whether a random process $Y$
{causally} influences another process $X$ or not. These results give a
natural interpretation of directed information $I(Y^n \to X^n)$ as the
amount of information that a random sequence $Y^n = (Y_1,Y_2,\ldots,
Y_n)$ {causally} provides about another random sequence $X^n =
(X_1,X_2,\ldots,X_n)$. A new measure, {\em directed lautum
  information}, is also introduced and interpreted in portfolio
theory, data compression, and hypothesis testing.
\end{abstract}


\section{Introduction}
Mutual information $I(X;Y)$ between two random variables $X$ and
$Y$arises as the canonical answer to a variety of questions in science
and engineering. Most notably, Shannon~\cite{Shannon48} showed that
the capacity $C$, the maximal data rate for reliable communication, of
a discrete memoryless channel $p(y|x)$ with input $X$ and output $Y$
is given by
\begin{equation}
\label{eq:shannon}
C = \max_{p(x)} I(X;Y).
\end{equation}
Shannon's channel coding theorem leads naturally to the operational
interpretation of mutual information $I(X;Y) = H(X) - H(X|Y)$ as the
amount of uncertainty about $X$ that can be reduced by observation
$Y$, or equivalently, the amount of information that $Y$ can provide
about $X$. Indeed, mutual information $I(X;Y)$ plays a central role in
Shannon's random coding argument, because the probability that
independently drawn sequences $X^n = (X_1,X_2,\ldots, X_n)$ and $Y^n =
(Y_1,Y_2,\ldots, Y_n)$ ``look'' as if they were drawn jointly decays
exponentially with $I(X;Y)$ in the first order of the
exponent. Shannon also proved a dual result \cite{Shannon60} that the
rate distortion function $R(D)$, the minimum compression rate to
describe a source $X$ by its reconstruction $\hat{X}$ within average
distortion $D$, is given by $R(D) = \min_{p(\hat{x}|x)}
I(X;\hat{X})$. In another duality result (the Lagrange duality this
time) to \eqref{eq:shannon}, Gallager \cite{Gal77} proved the minimax
redundancy theorem, connecting the redundancy of the universal
lossless source code to the maximum mutual information (capacity) of
the channel with conditional distribution that consists of the set of
possible source distributions (cf. \cite{Ryabko79}).

It has been shown that mutual information has also an important
role in problems that are not necessarily related to describing
sources or transferring information through channels. Perhaps the most
lucrative of such examples is the use of mutual information in gambling.
In 1956, Kelly~\cite{Kelly56} showed that if a horse race outcome can
be represented as an independent and identically distributed (i.i.d.)
random variable $X$, and the gambler has some side information $Y$
relevant to the outcome of the race, then
the mutual information $I(X;Y)$ captures the difference between growth
rates of the optimal gambler's wealth with and without side
information $Y$. Thus, Kelly's result gives an interpretation of
mutual information $I(X;Y)$ as the financial value of side information
$Y$ for gambling in the horse race $X$.

In order to tackle problems arising in information systems with
causally dependent components, Massey~\cite{Massey90} introduced the
notion of directed information, defined as
\begin{equation}\label{e_def_directed}
I(X^n\to Y^n) := \sum_{i=1}^n I(X^i;Y_i|Y^{i-1}),
\end{equation}
and showed that the normalized maximum directed information upper
bounds the capacity of channels with feedback.  Subsequently, it was
shown that Massey's directed information and its variants indeed
characterize the capacity of feedback and two-way
channels~\cite{Kramer98, Kramer03, PermuterWeissmanGoldsmith09,
  TatikondaMitter_IT09, Kim07_feedback, PermuterWeissmanChenMAC_IT09,
  ShraderPemuter07ISIT} and the rate distortion function with
feedforward \cite{Pradhan07Venkataramanan}.  Note that directed
information \eqref{e_def_directed} can be rewritten as
\begin{equation}\label{e_kim_identity}
I(X^n \to Y^n) = \sum_{i=1}^n I(X_i; Y_i^n | X^{i-1}, Y^{i-1}),
\end{equation}
each term of which corresponds to the achievable rate at time $i$
given side information $(X^{i-1}, Y^{i-1})$ (refer to
\cite{Kim07_feedback} for the details).

The main contribution of this paper is showing that directed
information has a natural interpretation in portfolio theory,
compression, and statistics when causality constraints exist. In stock
market investment (Sec.~\ref{s_portfolio_throey}), directed
information between the stock price $X$ and side information $Y$ is an
upper bound on the increase in growth rates due to {\em causal} side
information. This upper bound is tight when specialized to gambling in
horse races. In data compression (Sec.~\ref{s_data_compreession}) we
show that directed information characterizes the value of {causal}
side information in instantaneous compression, and it quantifies the
role of {causal} inference in joint compression of two stochastic
processes. In hypothesis testing (Sec.~\ref{s_hypothesis_testing}) we
show that directed information is the exponent of the minimum type~II
error probability when one is to decide if $Y_i$ has a {causal
  influence} on $X_i$ or not. Finally, we introduce the notion of
directed Lautum\footnote{{\it Lautum }(``elegant" in Latin) is the
  reverse spelling of ``mutual'' as aptly coined
  in~\cite{Palomar_verdu08LautumInformation}.}  information
(Sec.~\ref{s_lautum}), which is a causal extension of the notion of
Lautum information introduced by Palomar and Verd{\'u}
\cite{Palomar_verdu08LautumInformation}. We briefly discuss its role
in horse race gambling, data compression, and hypothesis testing.

\section{Preliminaries: Directed Information and Causal Conditioning}
\label{sec: directed information}

Throughout this paper, we use the {\it causal conditioning} notation
$(\cdot||\cdot)$ developed by Kramer~\cite{Kramer98}. We denote by
$p(x^n||y^{n-d})$ the probability mass function of $X^n = (X_1,\ldots,
X_n)$ \emph{causally conditioned} on $Y^{n-d}$ for some integer $d\geq
0$, which is defined as
\begin{equation} \label{e_causal_cond_def}
p(x^n||y^{n-d}) := \prod_{i=1}^{n} p(x_i|x^{i-1},y^{i-d}).
\end{equation}
(By convention, if $i < d$, then $x^{i-d}$ is set to null.)  In
particular, we use extensively the cases $d=0,1$:
\begin{align}
p(x^n||y^{n}) &:= \prod_{i=1}^{n}p(x_i|x^{i-1},y^{i}),\\
p(x^n||y^{n-1}) &:= \prod_{i=1}^{n} p(x_i|x^{i-1},y^{i-1}).
\end{align}
Using the chain rule, we can easily verify that
\begin{equation}\label{e_chain_rule}
p(x^n,y^n)=p(x^n||y^{n})p(y^n||x^{n-1}).
\end{equation}

The \emph{causally conditional} entropy $H(X^n||Y^n)$ and
$H(X^n||Y^{n-1})$ are defined respectively as
\begin{align}
H(X^n||Y^n) &:= E[\log p(X^n||Y^n)]=\sum_{i=1}^n H(X_i|X^{i-1},Y^i),
\nonumber\\
H(X^n||Y^{n-1}) &:= E[\log p(X^n||Y^{n-1})]=\sum_{i=1}^n H(X_i|X^{i-1},Y^{i-1}).
\label{e_causal_cond_entropy}
\end{align}
Under this notation, directed information defined in
(\ref{e_def_directed}) can be rewritten as
\begin{equation}\label{e_def_dir_diff_entrop_alt}
I(Y^n \to X^n)= H(X^n) - H(X^n||Y^n),
\end{equation}
which hints, in a rough analogy to mutual information, a possible
interpretation of directed information $I(Y^n \to X^n)$ as the amount
of information that causally available side information $Y^n$ can provide
about $X^n$.

Note that the channel capacity results \cite{Massey90,Kramer98,
  Kramer03, PermuterWeissmanGoldsmith09, TatikondaMitter_IT09,
  Kim07_feedback, PermuterWeissmanChenMAC_IT09, ShraderPemuter07ISIT}
involve the term $I(X^n\to Y^n)$, which measures the amount of
information transfer over the forward link from $X^n$ to $Y^n$.  In
gambling, however, the increase in growth rate is due to side
information (the backward link), and therefore the expression
$I(Y^n\to X^n)$ appears. Throughout the paper we also use the notation
$I(Y^{n-1}\to X^n)$ which denotes the directed information from the
vector $(\emptyset, Y^{n-1})$, i.e., the null symbol followed by
$Y^{n-1}$, to the vector to $X^n$, that is,\[
I(Y^{n-1} \to X^n) = \sum_{i=2}^n I(Y^{i-1}; X_i| X^{i-1}).
\]
Using the causal conditioning notation, given in
(\ref{e_causal_cond_entropy}), the directed information $I(Y^{n-1}\to
X^n)$ can be written as
 \begin{equation}\label{e_def_dir_diff_entrop2}
I(Y^{n-1} \to X^n)= H(X^n) - H(X^n||Y^{n-1}).
\end{equation}
 Directed information (in both directions) and mutual information obey
 the following conservation
 law
\begin{equation}\label{e_conservation_law} I(X^n; Y^n) =
   I(X^n\to Y^n) + I(Y^{n-1}\to X^n),
\end{equation}
which was shown by Massey and Massey\cite{Massey05}. The conservation
law is a direct consequence of the chain rule (\ref{e_chain_rule}),
and we show later in Sec.~\ref{s_cost_mismatch} that it has a natural
interpretation as a conservation of a mismatch cost in data
compression.

The causally conditional entropy rate of a random process $X$ given
another random process $Y$ and the directed information rate from $X$
to $Y$ are defined respectively as
\begin{equation}\label{e_def_entropy_rate}
\mathcal H(X||Y) := \lim_{n\to\infty} \frac{H(X^n||Y^n)}{n},
\end{equation}
\begin{equation}\label{e_def_directed_rate}
\mathcal I(X \to Y) := \lim_{n\to\infty} \frac{I(X^n \to Y^n)}{n},
\end{equation}
when these limits exist.
In particular, when $(X,Y)$ is stationary ergodic, both quantities
are well-defined, namely, the limits in (\ref{e_def_entropy_rate})
and (\ref{e_def_directed_rate}) exist \cite[Properties 3.5 and
3.6]{Kramer98}.

\section{Portfolio theory}\label{s_portfolio_throey}
Here we show that directed information is an upper bound on the
increment in growth rates of optimal portfolios in a stock market
due to {causal} side information. We start by considering a special
case where the market is a horse race gambling and show that the
upper bound is tight. Then we consider a general stock market
investment.

\subsection{Horse Race Gambling with Causal Side Information}
\label{s_horse_race}
Assume that there are $m$ racing horses and let $X_i$ denote the
winning horse at time $i$, i.e., $X_i \in \mathcal{X} :=
\{1,2,\ldots,m\}$.  At time $i$, the gambler has some side information
which we denote as $Y_i$. We assume that the gambler invests all
his/her capital in the horse race as a function of the previous horse
race outcomes $X^{i-1}$ and side information $Y^i$ up to time $i$. Let
$b(x_i|x^{i-1},y^{i})$ be the portion of wealth that the gambler bets
on horse $x_i$ given $X^{i-1} = x^{i-1}$ and $Y^i = y^{i}$. Obviously,
the gambling scheme should satisfy $b(x_i|x^{i-1},y^{i})\geq 0$ and
$\sum_{x_i\in \mathcal{X}} b(x_i|x^{i-1},y^{i})=1$ for any history
$(x^{i-1},y^i)$. Let $o(x_i|x^{i-1})$ denote the odds of a horse $x_i$
given the previous outcomes $x^{i-1}$, which is the amount of capital
that the gambler gets for each unit capital that the gambler invested
in the horse.  We denote by $S(x^n||y^n)$ the gambler's wealth after
$n$ races with outcomes $x^n$ and causal side information
$y^n$. Finally, $\frac{1}{n} W(X^n||Y^n)$ denotes the {\it growth
  rate} of wealth, where the growth $W(X^n||Y^n)$ is defined as the
expectation over the logarithm (base 2) of the gambler wealth, i.e.,
\begin{equation}\label{e_growth_rate_def}
W(X^n||Y^n):= E[\log S(X^n||Y^n)].
\end{equation}

Without loss of generality, we assume that the gambler's initial
wealth $S_0$ is 1. We assume that at any time $n$ the gambler invests
all his/her capital and therefore we have
\begin{equation}
S(X^n||Y^n)=b(X_{n}|X^{n-1},Y^{n})o(X_n|X^{n-1})S(X^{n-1}||Y^{n-1}).
\end{equation}
This also implies that
\begin{equation}
S(X^n||Y^n)=\prod_{i=1}^n b(X_{i}|X^{i-1},Y^{i})o(X_i|X^{i-1}).
\end{equation}
The following theorem establishes the investment strategy for
maximizing the growth.
\begin{theorem}[Optimal causal gambling]\label{t_gamble_all_money}
For any finite horizon $n$, the maximum growth is achieved when
the gambler invests the money proportional to the causally conditional
distribution of the horse race outcome, i.e.,
\begin{equation}\label{e_b=p}
b(x_i|x^{i-1},y^{i})=p(x_i|x^{i-1},y^{i})\quad \forall
i\in\{1,...,n\}, x^i\in \mathcal X^i, y^i\in \mathcal Y^i,
\end{equation}
and the maximum growth, denoted by $W^*(X^n||Y^n)$, is
\begin{equation}\label{e_growth_res}
W^*(X^n||Y^n):=\max_{\{b(x_i|x^{i-1},y^{i-1})\}_{i=1}^n}W(X^n||Y^n)= E[\log
o(X^n)]-H(X^n||Y^n).
\end{equation}
\end{theorem}

\medskip

Note that since $\{p(x_i|x^{i-1},y^{i})\}_{i=1}^n$ uniquely
determines $p(x^n||y^{n})$, and since
$\{b(x_i|x^{i-1},y^{i})\}_{i=1}^n$ uniquely determines
$b(x^n||y^{n})$, then (\ref{e_b=p}) is equivalent to
\begin{equation}
b(x^n||y^n)\equiv p(x^n||y^n).
\end{equation}

\begin{proof}
Consider
\begin{align}
W^*(X^n||Y^n)
&=\max_{b(x^{n}||y^n)} E[\log
b(X^n||Y^n)o(X^n)]\nonumber \\
&= \max_{b(x^{n}||y^n)} E[\log
b(X^n||Y^n)]+  E[\log o(X^n)]\nonumber \\
&=-H(X^n||Y^n)+ E[\log o(X^n)].
\end{align}
Here the last equality is achieved by choosing
$b(x^{n}||y^n)=p(x^{n}||y^n),$ and it is justified by the following
upper bound:
\begin{align}
E[\log b(X^n||Y^n)]
&=\sum_{x^n,y^n
}p(x^n,y^n)\left[ \log p(x^n||y^n)+\log\frac{b(x^n||y^n)}{p(x^n||y^n)}\right] \nonumber \\
&=-H(X^n||Y^n)+\sum_{x^n,y^n }p(x^n,y^n)\log\frac{b(x^n||y^n)}{p(x^n||y^n)}\nonumber \\
&\stackrel{(a)}{\leq}-H(X^n||Y^n)+\log\sum_{x^n,y^n
}p(x^n,y^n)\frac{b(x^n||y^n)}{p(x^n||y^n)}\nonumber \\
&\stackrel{(b)}{=}-H(X^n||Y^n)+\log\sum_{x^n,y^n }p(y^n||x^{n-1})b(x^n||y^n)\nonumber \\
&\stackrel{(c)}=-H(X^n||Y^n).
\end{align}
where (a) follows from Jensen's inequality, (b) from the chain rule, and (c) from the fact
that $\sum_{x^n,y^n }p(y^n||x^{n-1})b(x^n||y^n)=1$.
\end{proof}

In  case that the odds are fair, i.e., ${o}(X_i|X^{i-1})=1/m$,
\begin{equation}
W^*(X^n||Y^n)= n\log m-H(X^n||Y^n),
\end{equation}
and thus the sum of the growth rate and the entropy of the horse
race process conditioned causally on the side information is
constant, and one can see a duality between $H(X^n||Y^n)$ and
$W^*(X^n||Y^n)$.

Let us define $\Delta W(X^n||Y^n)$ as the increase in the growth due to causal side information, i.e.,
\begin{equation}
\Delta W(X^n||Y^n)=W^*(X^n||Y^n)-W^*(X^n).
\end{equation}

\begin{corollary}[Increase in the growth rate] \label{c_increase_double_rate}
The increase in growth rate due to the causal side information
sequence $Y^n$ for a horse race sequence $X^n$ is
\begin{equation}
\frac{1}{n}\Delta W(X^n||Y^n)=\frac{1}{n}I(Y^n\to X^n).
\end{equation}
\end{corollary}

As a special case, if the horse race outcome and side information are
pairwise i.i.d., then the (normalized) directed information
$\frac{1}{n}I(Y^n\to X^n)$ becomes the single letter mutual
information $I(X;Y)$, which coincides with Kelly's
result~\cite{Kelly56}.
\begin{proof}
From the definition of directed information
(\ref{e_def_dir_diff_entrop_alt}) and Theorem \ref{t_gamble_all_money} we
obtain
\begin{align*}
W^*(X^n||Y^n)-W^*(X^n)&= -H(X^n||Y^n)+H(X^n)=
I(Y^n\to X^n).
\end{align*}
\end{proof}

\begin{example}[Gambling in a Markov horse race process with causal side
information]\label{ex1}
Consider the case in which two horses are racing, and the winning
horse $X_i$ behaves as a Markov process as shown in
Fig.~\ref{f_Markov}. A horse that won will win again with probability
$p$ and lose with probability $1-p$ ($0\leq p\leq 1$). At time zero, we assume that the
two horses have equal probability of wining.  The side information
revealed to the gambler at time $i$ is $Y_i$, which is a noisy
observation of the horse race outcome $X_i$. It has probability $1-q$
of being equal to $X_i$, and probability $q$ of being different from
$X_i$. In other words, $Y_i=X_i+ V_i \mod 2$, where $V_i$ is a
Bernoulli($q$) process.

\begin{figure}[h!]{\footnotesize
\psfrag{Horse}[][][1]{Horse} \psfrag{Wins}[][][1]{wins}
\psfrag{1}[][][1]{$1$} \psfrag{2}[][][1]{$2$}
\psfrag{P}[][][1]{$p$} \psfrag{a}[][][1]{$1\!-\!p$}
\psfrag{Y1}[][][1]{$Y^n$} \psfrag{G1}[][][1]{}
\psfrag{G2}[][][0.9]{$X\!=\!1$} \psfrag{G3}[][][1]{$$}
\psfrag{T1}[][][1]{Horse 1 wins}
\psfrag{H2}[][][0.9]{$X\!=\!2$}
\psfrag{T2}[][][1]{Horse 2 wins}
\psfrag{c}[][][1]{1}
\psfrag{d}[][][1]{2}
\psfrag{X}[][][1]{$X$} \psfrag{Y}[][][1]{$Y$} \psfrag{T3}[][][1]{}
\psfrag{T4}[][][1]{} \psfrag{q}[][][1]{$q$}
\psfrag{q2}[][][1]{$1-q$}

\centerline{\includegraphics[width=6.6cm]{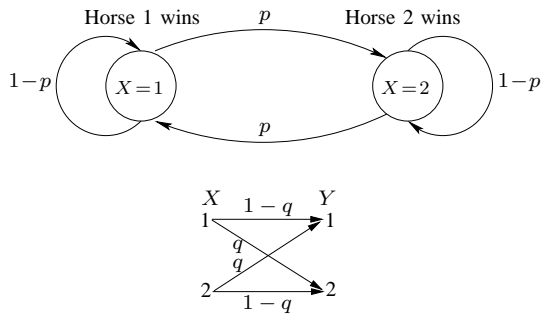}}
\caption{The setting of Example 1. The winning horse $X_i$ is
represented as a Markov process with two states. In state 1, horse
number 1 wins, and in state 2, horse number 2 wins. The side
information, $Y_i$, is a noisy observation of the winning horse,
$X_i$. } \label{f_Markov} }
\end{figure}

For this example, the increase in growth rate due to side
information $\Delta W:=\frac{1}{n}\Delta W(X^n||Y^n)$ is
\begin{align}\label{e_markov_qp}
\Delta W&=  h(p*q)-h(q),
\end{align}
where
$h(x):=-x\log x-(1-x)\log (1-x)$ is the binary entropy function, and $p*q=(1-p)q+(1-q)p$ denotes the parameter
of a Bernoulli distribution that results from convolving two
Bernoulli distributions with parameters $p$ and $q$.
\end{example}

The increase $\Delta W$ in the growth rate can be readily derived
using the identity in (\ref{e_kim_identity}) as follows:
\begin{align}\label{e_kim}
\frac{1}{n} I(Y^n\to X^n)
&\stackrel{(a)}{=} \frac{1}{n}\sum_{i=1}^n
I(Y_i;X_i^n|X^{i-1},Y^{i-1}) \nonumber \\
&\stackrel{(b)}{=} H(Y_1|X_{0})-H(Y_1|X_1),
\end{align}
where equality (a) is the identity from (\ref{e_kim_identity}), which can be easily verified by the chain rule for mutual
information \cite[eq.~(9)]{Kim07_feedback}, and (b) is due to the
stationarity of the process.

If the side information is known with some lookahead
$k\geq 0$, meaning that at time $i$ the gambler knows
$Y^{i+k}$, then the increase in growth rate is given by
\begin{align}
\Delta W&= \lim_{n\to \infty} \frac{1}{n} I(Y^{n+k}\to X^n) \nonumber \\
&\stackrel{}{=} H(Y_{k+1}|Y^{k},X_{0})-H(Y_1|X_1),
\end{align}
where the last equality is due to the same arguments as (\ref{e_kim}).
If the entire side information sequence $(Y_1,Y_2,\ldots)$ is known to
the gambler ahead of time, then since the sequence
$H(Y_{k+1}|Y^{k-1},X_{0})$ converges to the entropy rate of the
process, we obtain mutual information~\cite{Kelly56} instead of
directed information, i.e.,
\begin{align}
\Delta W&= \lim_{n\to \infty} \frac{1}{n} I(Y^{n};X^n)\nonumber\\
&=\lim_{n\to \infty} \frac{H(Y^n)}{n} -H(Y_1|X_1).
\end{align}

\subsection{Investment in a Stock Market with Causal Side Information}
We  use notation similar to the one in \cite[ch. 16]{CovThom06}. A stock
market at time $i$ is represented by a vector ${\bf
  X_i}=(X_{i1},X_{i2},\ldots,X_{im})$, where $m$ is the number of
stocks, and the {\it price relative $X_{ik}$} is the ratio of the
price of stock-$k$ at the end of day $i$ to the price of stock-$k$ at
the beginning of day $i$.  Note that gambling in a horse race is an
extreme case of stock market investment---for horse races, the price
relatives are all zero except one.

We assume that at time $i$ there is side information $Y^i$ that is
known to the investor.  A {\it portfolio} is an allocation of wealth
across the stocks.  A nonanticipating or causal portfolio strategy with
causal side information at time $i$ is denoted as ${\bf b}({\bf
  x}^{i-1},y^i)$, and it satisfies $\sum_{k=1}^m b_{k}({\bf
  x}^{i-1},y^i) =1$ and $b_{k}({\bf x}^{i-1},y^i)\geq 0$ for all
possible $({\bf x}^{i-1},y^i)$. We define $S({\bf x}^n||y^n)$ to be
the wealth at the end of day $n$ for a stock sequence ${\bf x}^n$ and
causal side information $y^n$. We have
\begin{equation}
S({\bf x}^n||y^n) = \left ({\bf b}^t({\bf x}^{n-1},y^n)\cdot  {\bf
x}_n\right ) S({\bf x}^{n-1}||y^{n-1}),
\end{equation}
where ${\bf b}^t \cdot{\bf x}$ denotes inner product between the two (column)
vectors ${\bf b}$ and ${\bf x}$.
The goal is to maximize the growth
\begin{equation}
W({\bf X}^n||Y^n):=E[\log S({\bf X}^n||Y^n) ].
\end{equation}
The justifications for maximizing the growth rate is due to
\cite[Theorem 5]{Algoet88} that such a portfolio strategy will exceed
the wealth of any other strategy to the first order in the exponent
for almost every sequence of outcomes from the stock market, namely,
if ${S^*({\bf X}^n||Y^n)}$ is the wealth corresponding to the growth
rate optimal return, then
\begin{equation}
\limsup_{n} \frac{1}{n} \log\left(\frac{S({\bf X}^n||Y^n)}{S^*({\bf
X}^n||Y^n)}\right)  \leq 0\quad\text{a.s.}
\end{equation}

Let us define
\begin{equation}
W({\bf X}_n|{\bf X}^{n-1},Y^n):=E[\log ({\bf b}^t({\bf X}^{n-1},Y^n)
{\bf X}_n)].
\end{equation}
From this definition follows the chain rule:
\begin{equation}
W({\bf X}^n||Y^n)=\sum_{i=1}^n W({\bf X}_i|{\bf X}^{i-1},Y^i),
\end{equation}
from which we obtain
\begin{align}\label{e_maxWn_as_maxmize_each_one_market}
\max_{\{{\bf b}({\bf x}^{i-1},y^i)\}_{i=1}^n}W({\bf
X}^n||Y^n)&=\sum_{i=1}^n \max_{{\bf b}({\bf x}^{i-1},y^i)}
W({\bf X}_i|X^{i-1},Y^{i})\nonumber \\
&=\sum_{i=1}^n \int_{{\bf x}^{i-1},y^i}f({\bf
x}^{i-1},y^i)\max_{{\bf b}({\bf x}^{i-1},y^i)} W({\bf X}_i|{\bf
x}^{i-1},y^{i}),
\end{align}
where $f({\bf x}^{i-1},y^i)$ denotes the probability density function
of $({\bf x}^{i-1},y^i)$. The maximization in
(\ref{e_maxWn_as_maxmize_each_one_market}) is equivalent to the
maximization of the growth rate for the memoryless case where the
cumulative distribution function of the stock-vector ${\bf X}$ is
$P({\bf X}\leq {\bf x})=\Pr({\bf X}_i\leq {\bf x}|{\bf x}^{i-1},{
  y}^i)$ and the portfolio ${\bf b} = {\bf b}({\bf x}^{i-1},y^i)$ is a
function of $(x^{i-1}, y^i)$, i.e.,
\begin{align}\label{e_memoryless_opt}
\text{maximize } &
E[\log({\bf b}^t {\bf X})|X^{i-1} = x^{i-1}, Y^i = y^i]\nonumber \\
\text{subject to } &
 \sum_{i=1}^m b_k=1, \nonumber \\
&  b_k\geq0, \forall k\in [1,2,\ldots,m].
\end{align}

In order to upper bound the difference in growth rate due to {causal}
side information we recall the following result 
 which bounds the loss in growth rate incurred by optimizing
the portfolio with respect to a wrong distribution $g({\bf x})$ rather than
the true distribution $f({\bf x})$.

\begin{theorem}[\!{\cite[Theorem 1]{BarCov1988}}]\label{t_diff_stock}
Let $f({\bf x})$ be the probability density function of a stock vector
${\bf X}$, i.e., ${\bf X}\sim f({\bf x})$. Let ${\bf b}_f$ be the
growth rate portfolio corresponding to $f({\bf x})$, and let ${\bf
  b}_g$ be the growth rate portfolio corresponding to another density
$g({\bf x})$. Then the increase in optimal growth rate $\Delta W$ by
using ${\bf b}_f$ instead of ${\bf b}_g$ is bounded by
\begin{equation}
\Delta W= E[\log ({\bf b}^t_f {\bf X})]-E[\log ({\bf b}^t_g {\bf
X})]\leq D(f||g),
\end{equation}
where $D(f||g):=\int f(x)\log \frac{f(x)}{g(x)} dx$ denotes the
Kullback--Leibler divergence between the probability density
functions $f$ and $g$.
\end{theorem}

Using Theorem \ref{t_diff_stock}, we can upper bound the increase in
growth rate due to causal side information by directed information as
shown in the following theorem.
\begin{theorem}[Upper bound on increase in growth rate]
\label{t_stock_directed_upper_bound}
The increase in optimal growth rate for a stock market sequence ${\bf
  X}^n$ due to side information $Y^n$ is upper bounded by
\begin{equation}
W^*({\bf X}^n||Y^n)- W^*({\bf X}^n) \leq I(Y^n\to {\bf X}^n),
\end{equation}
where $W^*({\bf X}^n||Y^n)\triangleq \max_{\{{\bf b}({\bf
X}^{i-1},Y^i)\}_{i=1}^n}W({\bf X}^n||Y^n)$ and $W^*({\bf X}^n):=
\max_{\{{\bf b}({\bf X}^{i-1})\}_{i=1}^n}W({\bf X}^n)$.
\end{theorem}

\begin{proof} Consider
\begin{align}
\lefteqn{W^*({\bf X}^n||Y^n)- W^*({\bf X}^n)}\nonumber\\
&=\sum_{i=1}^n \int_{{\bf x}^{i-1},y^i}f({\bf x}^{i-1},y^i)\left[
\max_{{\bf b}({\bf x}^{i-1},y^i)} W({\bf X}_i|{\bf
x}^{i-1},y^{i})-\max_{{\bf b_i}({\bf x}^{i-1})} W({\bf X}_i|{\bf
x}^{i-1})\right] \nonumber \\
&\stackrel{(a)}{\leq}\sum_{i=1}^n \int_{{\bf x}^{i-1},y^i}f({\bf
x}^{i-1},y^i)\left[ \int_{{\bf x}_i} f({\bf x}_i|{\bf
x}^{i-1},y^i)\log \frac{f({\bf x}_i|{\bf x}^{i-1},y^i)}{f({\bf
x}_i|{\bf
x}^{i-1})} \right] \nonumber \\
&\stackrel{}{=}\sum_{i=1}^n E\left[ \log \frac{f({\bf X}_i|{\bf
X}^{i-1},Y^i)}{f({\bf X}_i|{\bf
X}^{i-1})} \right] \nonumber \\
&\stackrel{}{=}\sum_{i=1}^n h({\bf X}_i|{\bf X}^{i-1})-h({\bf
X}_i|{\bf
X}^{i-1},Y^i) \nonumber \\
&\stackrel{}{=}I(Y^n \to {\bf X}^n),
\end{align}
where the inequality (a) is due to Theorem \ref{t_diff_stock}.
\end{proof}

Note that the upper bound in
Theorem~\ref{t_stock_directed_upper_bound} is tight for gambling in
horse races (Corollary \ref{c_increase_double_rate}).

\section{Data Compression}\label{s_data_compreession}
In this section we investigate the role of directed information in
data compression and find two interpretations:
\begin{enumerate}
\item directed information characterizes the value of causal side
  information in instantaneous compression, and
\item it also quantifies the role of causal inference in joint
  compression of two stochastic processes.
\end{enumerate}

\subsection{Instantaneous Lossless Compression with Causal Side Information}
Let $X_1, X_2 \ldots$ be a source and $Y_1, Y_2, \ldots$ be side
information about the source. The source is to be encoded losslessly
by an instantaneous code with causally available side information, as
depicted in Fig. \ref{f_causal_compression}. More formally, an {\it
  instantaneous lossless source encoder with causal side information}
consists of a sequence of mappings $\{ M_i \}_{i \geq 1}$ such that
each $M_i : \mathcal{X}^i \times \mathcal{Y}^i \mapsto \{0,1\}^*$ has
the property that for every $x^{i-1}$ and $y^i$, $M_i (x^{i-1} \cdot ,
y^i)$ is an instantaneous (prefix) code.

\begin{figure}[h!]{
\psfrag{X}[B][][1]{$X^i$}
\psfrag{M}[B][][1]{$\;\;\;\;\;\;\;\;\;\;\;M_i(X^i,Y^i)$}
\psfrag{Causal}[][][1]{Causal} \psfrag{Encoder}[][][1]{Encoder}
\psfrag{Decoder}[][][1]{Decoder\;}
\psfrag{Y}[][][1]{$\;Y^i$}
\psfrag{H}[B][][1]{$\;\;\;\;\;\;\hat X_i(M^i,Y^i)$}
\centerline{\includegraphics[width=8cm]{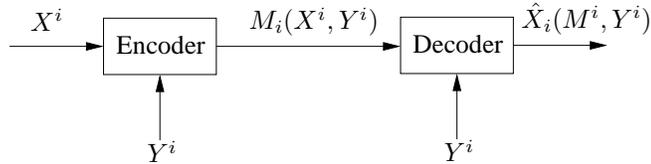}}
\caption{Instantaneous  data compression with causal side
information } \label{f_causal_compression} }
\end{figure}

An instantaneous lossless source encoder with causal side information
operates sequentially, emitting the concatenated bit stream $M_1 (X_1,
Y_1) M_2(X^2, Y^2) \ldots$. The defining property that $M_i (x^{i-1}
\cdot , y^i)$ is an instantaneous code for every $x^{i-1}$ and $y^i$
is a necessary and sufficient condition for the existence of a decoder
that can losslessly recover $x^i$ based on $y^i$ and the bit stream
$M_1 (x_1, y_1) M_2(x^2, y^2) \ldots$ as soon as it receives $M_1
(x_1, y_1) M_2(x^2, y^2) \ldots M_i(x^i, y^i)$ for all sequence pairs
$(x_1,y_1), ( x_2, y_2), \ldots$, and all $i \geq 1$. Let $L(x^n||
y^n)$ denote the length of the concatenated string $M_1 (x_1, y_1)
M_2(x^2, y^2) \ldots M_n(x^n, y^n)$. Then the following result is due
to Kraft's inequality and Huffman coding adapted to the case where
causal side information is available.
\begin{theorem}[Lossless source coding with causal side information]
\label{th: lossless insta coding with si}
Any instantaneous lossless source encoder with causal side
information satisfies
\begin{equation}\label{eq: lossless source encoder converse}
\frac{1}{n} E L(X^n|| Y^n) \geq \frac{1}{n} \sum_{i=1}^n H(X_i |
X^{i-1}, Y^i) \ \ \ \ \ \forall n \geq 1.
\end{equation}
There exists an instantaneous lossless source encoder with causal side
information satisfying
\begin{equation}\label{eq: direct part for inst lossless causal}
\frac{1}{n} E L(X^n|| Y^n) \leq  \frac{1}{n} \sum_{i=1}^n r_i + H(X_i
| X^{i-1}, Y^i)  \ \ \ \ \ \forall i \geq 1,
\end{equation}
where $r_i= \sum_{x^{i-1},y^{i}}p(x^{i-1},y^{i})\min(1, \max_{x_i}p(x_i|x^{i-1},y^{i-1})+0.086)$.
\end{theorem}
\begin{proof}
The lower bound follows from Kraft's inequality \cite[Theorem
  5.3.1]{CovThom06} and the upper bound follows from Huffman coding on
the conditional probability $p(x_i|x^{i-1},y^{i})$.  The redundancy
term $r_i$ follows from Gallager's redundancy bound
\cite{Gallager78HuffmanCode}, $\min(1, P_{i}+0.086)$, where $P_{i}$ is
the probability of the most likely source letter at time $i$,
averaged over side information sequence $(X^{i-1},Y^{i})$.
\end{proof}

Since the Huffman code achieves the entropy rate for dyadic
probability, it follows that if the conditional probability
$p(x_i|x^{i-1},y^{i-1})$ is dyadic, i.e., if each conditional
probability equals to $2^{-k}$ for some integer $k$, then (\ref{eq:
  lossless source encoder converse}) can be achieved with equality.

Theorem \ref{th: lossless insta coding with si}, combined with the
identity $\sum_{j=i}^n H(X_i | X^{i-1}, Y^i) = H(X^n) - I(Y^n
\rightarrow X^n)$, implies that the compression rate saved in optimal
sequential lossless compression due to the causal side information is
upper bounded by $\frac{1}{n}I(Y^n\to X^n)-1$, and lower bounded by
$\frac{1}{n}I(Y^n\to X^n)+1$. If all the probabilities are dyadic,
then the compression rate saving is exactly the directed information
rate $\frac{1}{n}I(Y^n\to X^n)$. This saving should be compared to
$\frac{1}{n}I(X^n;Y^n)$, which is the saving in the absence of
causality constraint.

\subsection{Cost of Mismatch in Data Compression}
\label{s_cost_mismatch}
Suppose we compress a pair of correlated sources $\{(X_i,Y_i)\}$
jointly with an optimal lossless variable length code (such as the
Huffman code), and we denote by $E(L(X^n,Y^n))$ the average length of
the code. Assume further that $Y_i$ is generated randomly by a forward
link $p(y_i|y^{i-1},x^{i})$ as in a communication channel or a
chemical reaction, and $X_i$ is generated by a backward link
$p(x_i|y^{i-1},x^{i-1})$ such as in the case of an encoder or a
controller with feedback. By the chain rule for causally conditional
probabilities (\ref{e_chain_rule}), any joint distribution can be
modeled according to Fig.~\ref{f_compression}.
\begin{figure}[h!]{
\psfrag{A1}[][][1]{$p(x_i|x^{i-1},y^{i-1})$}
\psfrag{A2}[][][1]{Backward link}
\psfrag{B1}[][][1]{$p(y_i|x^{i},y^{i-1})$}
\psfrag{B3}[][][1]{$p(y_i|y^{i-1})$}  \psfrag{B2}[][][1]{Forward
link} \psfrag{X}[][][1]{$X_i$} \psfrag{Y}[][][1]{$Y_i$}
\psfrag{Y2}[][][1]{$Y_{i-1}$} \psfrag{H1}[][][1]{}
\psfrag{H2}[][][1]{}
\centerline{\includegraphics[width=8cm]{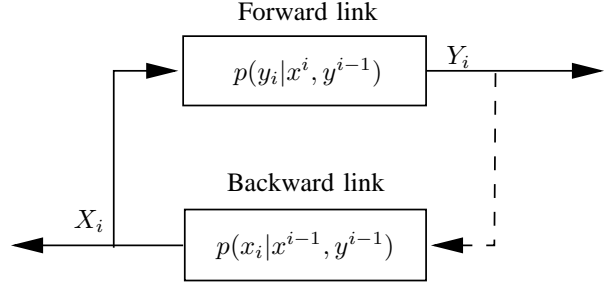}}
\caption{ {Compression of two correlated sources $\{X_i,Y_i\}_{i\geq
1 }$. Since any joint distribution can be decomposed as
$p(x^n,y^n)=p(x^n||y^{n-1})p(y^n||x^n)$, each link embraces the
existence of a forward or feedback channel (chemical reaction). We
investigate the influence of the link knowledge on joint
compression of $\{X_i,Y_i\}_{i\geq 1 }$}. \label{f_compression}} }
\end{figure}

Recall that the optimal variable-length lossless code, in which both
links are taken into account, has the average length
\[
H(X^n,Y^n)\le E(L(X^n,Y^n)) < H(X^n,Y^n) + 1.
\]
However, if the code is erroneously designed to be optimal for the
case in which the forward link does not exist, namely, the code is
designed for the joint distribution $p(y^n)p(x^n||y^{n-1})$, then the
average code length (up to 1 bit) is
\begin{align}
E(L(X^n,Y^n))&=\sum_{x^n,y^n}p(x^n,y^n)\log \frac{1}{p(y^n)p(x^n||y^n-1)}\nonumber \\
&= \sum_{x^n,y^n}p(x^n,y^n)\log \frac{p(x^n,y^n)}{p(y^n)p(x^n||y^n-1)}+H(X^n,Y^n)\nonumber \\
&=\sum_{x^n,y^n}p(x^n,y^n)\log \frac{p(y^n||x^n)}{p(y^n)}+H(X^n,Y^n)\nonumber \\
&=I(X^n\to Y^n)+H(X^n,Y^n).
\end{align}
Hence the redundancy (the gap from the minimum average code length)
is $I(X^n\to Y^n)$. Similarly, if the backward
link is ignored, then the average code length (up to 1 bit) is
\begin{align}
E(L(X^n,Y^n))&=\sum_{x^n,y^n}p(x^n,y^n)\log
\frac{1}{p(y^n||x^n)p(x^n)}\nonumber \\ &=
\sum_{x^n,y^n}p(x^n,y^n)\log
\frac{p(x^n,y^n)}{p(y^n||x^n)p(x^n)}+H(X^n,Y^n)\nonumber
\\ &=\sum_{x^n,y^n}p(x^n,y^n)\log
\frac{p(x^n||y^{n-1})}{p(x^n)}+H(X^n,Y^n)\nonumber \\ &=I(Y^{n-1}\to
X^n)+H(X^n,Y^n)
\end{align}
Hence the redundancy for this case is $I(Y^{n-1}\to X^n)$. If both
links are ignored, the redundancy is simply the mutual information
$I(X^n; Y^n)$. This result quantifies the value of knowing causal
influence between two processes when designing the optimal joint
compression. Note that the redundancy due to ignoring both links is
the sum of the redundancies from ignoring each link. This recovers the
conservation law (\ref{e_conservation_law}) operationally.

\section{Directed Information and Statistics: Hypothesis Testing}
\label{s_hypothesis_testing}

Consider a system with an input sequence $(X_1,X_2,\ldots,X_n)$ and
output sequence $(Y_1,Y_2,\ldots,Y_n)$, where the input is generated
by a stimulation mechanism or a controller, which observes the
previous outputs, and the output may be generated either causally from
the input according to $\{p(y_i|y^{i-1},x^{i})\}_{i=1}^n$ (the null
hypothesis $H_0$) or independently from the input according to
$\{p(y_i|y^{i-1})\}_{i=1}^n$ (the alternative hypothesis $H_1$). For
instance, this setting occurs in communication or biological systems,
where we wish to test whether the observed system output $Y^n$ is in
response to one's own stimulation input $X^n$ or to some other input
that uses the same stimulation mechanism and therefore induces the
same marginal distribution $p(y^n)$.  The stimulation mechanism
$p(x^n||y^{n-1})$, the output generator $p(y^n||x^n)$, and the
sequences $X^n$ and $Y^n$ are assumed to be known.

\begin{figure}[h]{
\psfrag{A1}[][][1]{$p(x_i|x^{i-1},y^{i-1})$}
\psfrag{A2}[][][1]{Controller}
\psfrag{B1}[][][1]{$p(y_i|x^{i},y^{i-1})$}
\psfrag{B3}[][][1]{$p(y_i|y^{i-1})$}
\psfrag{B2}[][][1]{Output generator}
\psfrag{X}[][][1]{$X_i$}
\psfrag{Y}[][][1]{$Y_i$}
\psfrag{Y2}[][][1]{$Y_{i-1}$}
\psfrag{H1}[][][1]{Hypothesis $H_0$:}
\psfrag{H2}[][][1]{Hypothesis $H_1$}

\centerline{\includegraphics[width=16cm]{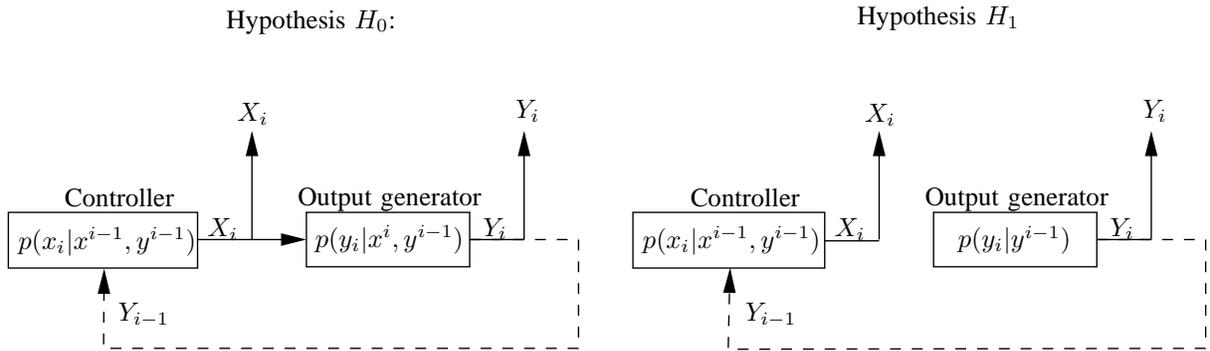}}

\caption{Hypothesis testing. $H_0$: The input sequence
  $(X_1,X_2\ldots,X_n)$ causally influences the output sequence
  $(Y_1,Y_2,\ldots,Y_n)$ through the causal conditioning distribution
  $p(y^n||x^n)$. $H_1$: The output sequence $(Y_1,Y_2,\ldots,Y_n)$ was
  not generated by the input sequence $(X_1,X_2,\ldots,X_n)$, but by
  another input from the same stimulation mechanism
  $p(x^n||y^{n-1})$. }
\label{f_hypothesis}
}\end{figure}

An {\it acceptance region} $A$ is the set of all sequences $(x^n,y^n)$
for which we accept the null hypothesis $H_0$. The complement of $A$,
denoted by $A^c$, is the rejection region, namely, the set of all
sequences $(x^n,y^n)$ for which we reject the null hypothesis $H_0$
and accept the alternative hypothesis $H_1$.  Let
\begin{equation}\label{e_error_def}
\alpha:=\Pr(A^c|H_0),\quad \beta:=\Pr(A|H_1)
\end{equation}
denote the probabilities of {\it type I error} and {\it type II
  error}, respectively.

The following theorem interprets the directed information rate
$\mathcal I(X\to Y)$ as the best error exponent of $\beta$
that can be achieved while $\alpha$ is less than some constant
$\epsilon> 0$.

\begin{theorem}[Chernoff--Stein Lemma for the causal dependence test: Type II error]
\label{t_Chernoff-Stel}
Let $(X,Y) = \{X_i,Y_i\}_{i=1}^\infty$ be a stationary and ergodic
random process.  Let $A_n\subseteq {(\mathcal X\times \mathcal Y)}^n$
be an acceptance region, and let $\alpha_n$ and $\beta_n$ be the
corresponding probabilities of type I and type II errors
(\ref{e_error_def}).  For $0<\epsilon<\frac{1}{2}$, let
\begin{equation}\label{e_beta_n}
\beta_n^{(\epsilon)}=\min_{A_n\subseteq {(\mathcal X\times \mathcal Y)}^n, \alpha_n<\epsilon} \beta_n.
\end{equation}
Then
\begin{equation}
\lim_{n\to\infty} -\frac{1}{n} \log \beta_n^{(\epsilon)}=\mathcal I(X\to Y),
\end{equation}
where the directed information rate is the one induced by the joint distribution from $H_0$, i.e., $p(x^n||y^{n-1})p(y^n||x^n)$.
\end{theorem}

Theorem \ref{t_Chernoff-Stel} is reminiscent of the achievability
proof in the channel coding theorem. In the random coding achievability proof
\cite[ch 7.7]{CovThom06} we check whether the output $Y^n$ is
resulting from a message (or equivalently from an input sequence
$X^n$) and we would like to have the error exponent which is, according to
Theorem \ref{t_Chernoff-Stel}, $I(X^n\to Y^n)$ to be as large as
possible so we can distinguish as many messages as possible.

The proof of Theorem \ref{t_Chernoff-Stel} combines arguments from
the Chernoff--Stein Lemma \cite[Theorem 11.8.3]{CovThom06} with the
Shannon--McMillan--Breiman Theorem for directed information~\cite[Lemma
  3.1]{Pradhan07Venkataramanan}, which implies that for a jointly
stationary ergodic random process
\begin{equation*}
\frac{1}{n}\log \frac{p(Y^n||X^n)}{P(Y^n)} \to \mathcal I( X\to Y)  \text{ in probability.}
\end{equation*}

\begin{proof}
{\em Achievability:}
Fix $\delta>0$ and let $A_n$ be
\begin{equation}\label{e_A_n}
A_n=\left\{x^n,y^n:\left|\log \frac{p(y^n||x^n)}{p(y^n)}-\mathcal I( X\to Y)\right|<\delta  \right\}
\end{equation}
By the AEP for directed information~\cite[Lemma
  3.1]{Pradhan07Venkataramanan} we have that $\Pr (A_n|H_0)\to 1$ in
probability; hence there exists $N(\epsilon)$ such that for all
$n>N(\epsilon)$, $\alpha_n=\Pr (A_n^c|H_0)< \epsilon$. Furthermore,
\begin{align}
\beta_n&=\Pr(A_n|H_1)\nonumber \\
&=\sum_{x^n,y^n\in A_n} p(x^n||y^{n-1})p(y^n)\nonumber \\
&\stackrel{(a)}{\leq}\sum_{x^n,y^n\in A_n} p(x^n||y^{n-1})p(y^n||x^n)2^{-n(\mathcal I(X\to Y)-\delta)}\nonumber \\
&\stackrel{}{=}2^{-n(\mathcal I(X\to Y)-\delta)}\sum_{x^n,y^n\in A_n} p(x^n||y^{n-1})p(y^n||x^n)\nonumber \\
&\stackrel{(b)}{=}2^{-n(\mathcal I(X\to Y)-\delta)}(1-\alpha_n),
\end{align}
where inequality (a) follows from the definition of $A_n$ and (b) from the definition of $\alpha_n$.
We conclude that
\begin{equation}\label{e_Bn_limit}
\lim \sup _{n\to\infty}\frac{1}{n}\log \beta_n\leq -\mathcal I(X\to Y)+\delta,
\end{equation}
establishing the achievability since $\delta>0$ is arbitrary.

{\em Converse:} Let $B_n\subseteq {(\mathcal X\times \mathcal Y)}^n$
such that $\Pr(B_n^c|H_0) < \epsilon < \frac{1}{2}$. Consider
\begin{align}
\Pr(B_n|H_1)&\geq \Pr(A_n\cap B_n|H_1) \nonumber \\
&= \sum_{(x^n,y^n) \in A_n\cap B_n} p(x^n||y^{n-1})p(y^n) \nonumber \\
&\geq \sum_{(x^n,y^n) \in A_n\cap B_n} p(x^n||y^{n-1})p(y^n||x^{n-1})2^{-n(\mathcal I(X\to Y)+\delta)} \nonumber \\
&= 2^{-n(\mathcal I(X\to Y)+\delta)}\Pr(A_n\cap B_n|H_0)\nonumber \\
&= 2^{-n(\mathcal I(X\to Y)+\delta)}(1-\Pr(A_n^c\cup B_n^c|H_0))\nonumber \\
&\geq 2^{-n(\mathcal I(X\to Y)+\delta)}(1-\Pr(A_n^c|H_0)-\Pr(B_n^c|H_0))
\end{align}
Since $\Pr(A_n^c|H_0)\to 0$ and
$\Pr(B_n^c|H_0)<\epsilon<\frac{1}{2}$,
we obtain
\begin{equation}
\lim \inf_{n\to\infty}\frac{1}{n}\log \beta_n\geq -(\mathcal I(X\to
Y)+\delta).
\end{equation}
Finally, since $\delta > 0$ is arbitrary, the proof of the converse is
completed.
\end{proof}

\section{Directed Lautum Information}\label{s_lautum}
Recently, Palomar and Verd{\'u}
\cite{Palomar_verdu08LautumInformation} have defined the lautum
information $L(X^n;Y^n)$ as
\begin{equation}
L(X^n;Y^n):=\sum_{x^n,y^n}p(x^n)p(y^n)\log \frac{p(y^n)}{p(y^n|x^n)},
\end{equation}
and showed that it has operational interpretations in statistics,
compression, gambling, and portfolio theory, where the true
distribution is $p(x^n)p(y^n)$ but mistakenly a joint distribution
$p(x^n,y^n)$ is assumed. As in the definition of directed information
wherein the role of regular conditioning is replaced by causal
conditioning, we define two types of directed lautum information. The
first type
\begin{align}
L_1(X^n\to Y^n) &:=\sum_{x^n,y^n}p(x^n)p(y^n)\log
\frac{p(y^n)}{p(y^n||x^n)}
\end{align}
and the second type
\begin{align}
L_2(X^n\to Y^n) &:=\sum_{x^n,y^n}p(x^n||y^{n-1})p(y^n)\log
\frac{p(y^n)}{p(y^n||x^n)}.
\end{align}
When $p(x^n||y^{n-1})=p(x^n)$ (no feedback), the two definitions
coincide.  We will see in this section that the first type of directed
lautum information has operational meanings in scenarios where the
true distribution is $p(x^n)p(y^n)$ and, mistakenly, a joint
distribution of the form $p(x^n)p(y^n||x^n)$ is assumed. Similarly,
the second type of directed information occurs when the true
distribution is $p(x^n||y^{n-1})p(y^n)$, but a joint distribution of
the form $p(x^n||y^{n-1})p(y^n||x^n)$ is assumed.

We have the following conservation law for the first-type directed
lautum information:
\begin{lemma}[Conservation law for the first type of directed latum information]
For any discrete jointly distributed random vectors $X^n$ and $Y^n$
\begin{equation}
L(X^n;Y^n)=L_1(X^n\to Y^n)+L_1(Y^{n-1}\to X^n).
\end{equation}
\end{lemma}

\begin{proof}
Consider
\begin{align}\label{e_cons_lautum}
L_1(X^n;Y^n)&\stackrel{(a)}{=}
\sum_{x^n,y^n}p(x^n)p(y^n)\log \frac{p(y^n)p(x^n)}{p(y^n,x^n)}\nonumber\\
&\stackrel{(b)}{=}\sum_{x^n,y^n}p(x^n)p(y^n)\log \frac{p(y^n)p(x^n)}{p(y^n||x^n)p(x^n||y^{n-1})}\nonumber\\
&\stackrel{}{=}\sum_{x^n,y^n}p(x^n)p(y^n)\log \frac{p(y^n)}{p(y^n||x^n)}+\sum_{x^n,y^n}p(x^n)p(y^n)\log \frac{p(x^n)}{p(x^n||y^{n-1})}\nonumber\\
&\stackrel{}{=} L_1(X^n\to Y^n)+L_1(Y^{n-1}\to X^n),
\end{align}
where (a) follows from the definition of lautum information and (b)
follows from the chain rule $p(y^n,x^n)=p(y^n||x^n)p(x^n||y^{n-1})$.
\end{proof}

A direct consequence of the lemma is the following  condition for
the equality between two types of directed lautum information and
regular lautum information.

\begin{corollary}
If
\begin{equation}
L(X^n;Y^n)=L_1(X^n\to Y^n),
\end{equation}
then
\begin{equation}\label{e_cond}
p(x^n)=p(x^n||y^{n-1})\text{ for all }(x^n,y^n) \in \mathcal
X^n\times \mathcal Y^n \text{ with } p(x^n,y^n)>0.
\end{equation}
Conversely, if (\ref{e_cond}) holds, then
\begin{equation}
L(X^n;Y^n)=L_1(X^n\to Y^n)=L_2(X^n\to Y^n).
\end{equation}

\end{corollary}

\begin{proof}
The proof of the first part follows from the conservation
law (\ref{e_cons_lautum}) and the nonnegativity of Kullback--Leibler
divergence~\cite[Theorem 2.6.3]{CovThom06} (i.e., $L_1(Y^{n-1}\to
X^n)=0$ implies that $p(x^n)=p(x^n||y^{n-1})$). The second part
follows from the definitions of regular and directed lautum
information.
\end{proof}

The {\it lautum information rate} and {\it directed lautum
  information rates} are respectively defined as
 \begin{align}
\mathcal L( X; Y)&:=\lim_{n\to \infty} \frac{1}{n}L( Y^n ; X^n),\label{e_lautum_rate} \\
\mathcal L_j(X \to Y)&:=\lim_{n\to \infty} \frac{1}{n}L_j( Y^n \to
X^n)\quad \text{for } j=1,2\label{e_directed_lautum_rate},
\end{align}
whenever the limits exist. The next lemma provides a technical
condition for the existence of the limits.

\begin{lemma} \label{l_suff_lautum_rate}
If the process $(X_n,Y_n)$ is stationary and Markov (i.e.,
$p(x_i,y_i|x^{i-1},y^{i-1})=p(x_i,y_i|x_{i-k}^{i-1},y_{i-k}^{i-1})$
for some finite $k$), then $\mathcal L(X; Y)$ and $\mathcal L_2(X \to
Y)$ are well defined.  Similarly, if the process $\{X^n,Y^n\}\sim
p(x^n)p(y^n||x^n)$ is stationary and Markov, then
$\mathcal L_1(X \to Y)$ is well defined.
 \end{lemma}

\begin{proof}
It is easy to see the sufficiency of the conditions for $\mathcal
L_1( X \to \ Y)$ from the following identity:
 \begin{align*}
 L_1(\mathcal X^n \to \mathcal Y^n)
&=
\sum_{x^n,y^n}p(x^n)p(y^n)\log \frac{p(x^n)p(y^n)}{p(x^n)p(y^n||x^n)}\nonumber \\
&=-H(X^n)-H(Y^n)- \sum_{x^n,y^n}p(x^n)p(y^n)\log {p(x^n)p(y^n||x^n)}.
\end{align*}
Since the process is stationary the limits $\lim_{n\to \infty
}\frac{1}{n}H(X^n)$ and $\lim_{n\to \infty } \frac{1}{n}H(Y^n)$
exist. Furthermore, since $p(x^n,y^n)=p(x^n)p(y^n||x^n)$ is assumed
to be stationary and  Markov, the limit $\lim_{n\to
\infty } \frac{1}{n}\sum_{x^n,y^n}p(x^n)p(y^n)\log
{p(x^n)p(y^n||x^n)}$ exists.  The sufficiency of the condition
can be proved for $\mathcal L_2( X \to  Y)$ and the lautum
information rate using a similar argument.
\end{proof}

Adding causality constraints to the problems that were considered in
\cite{Palomar_verdu08LautumInformation}, we obtain the following
results for horse race gambling and data compression.

\subsection{Horse Race Gambling with Mismatched Causal Side Information}

Consider the horse race setting in Section \ref{s_horse_race} where
the gambler has causal side information. The joint distribution of
horse race outcomes $X^n$ and the side information $Y^n$ is given by
$p(x^n)p(y^n||x^{n-1})$, namely, $X_i \to X^{i-1} \to Y^i$ form a Markov chain,
and therefore the side information does not increase the
growth rate. The gambler mistakenly assumes a joint distribution
$p(x^n||y^n)p(y^n||x^{n-1})$,  and therefore he/she uses a
gambling scheme $b^*(x^n||y^n)=p(x^n||y^n)$.

\begin{theorem}
If the gambling scheme $b^*(x^n||y^n)=p(x^n||y^n)$ is applied to the
horse race described above, then the penalty in the growth with
respect to the gambling scheme $b^*(x^n)$ that uses no side
information is $L_2(Y^n\to X^n)$. For the special case where the
side information is independent of the horse race outcomes, the
penalty is $L_1(Y^n\to X^n)$.
\end{theorem}

\begin{proof}
The optimal growth rate where the joint distribution is
$p(x^n)p(y^n||x^{n-1})$ is $W^*(X^n)= E[\log o(X^n)]-H(X^n)$. Let
$E_{p(x^n)p(y^n||x^{n-1})}$ denotes the expectation with respect to
the joint distribution $p(x^n)p(y^n||x^{n-1})$. The growth rate for
the gambling strategy $b(x^n||y^n)=p(x^n||y^n)$ is
\begin{align}
W^*(X^n||Y^n)&\stackrel{}{=} E_{p(x^n)p(y^n||x^{n-1})}[\log
b(X^n||Y^n)o(X^n)]\nonumber \\
&=  E_{p(x^n)p(y^n||x^{n-1})}[\log
p(X^n||Y^n)]+  E_{p(x^n)p(y^n||x^{n-1})}[\log o(X^n)];
\end{align}
hence $W^*(X^n)-W^*(X^n||Y^n)=L_2(Y^n\to X^n)$. In the special case,
where  the side information is independent of the horse outcome,
namely, $p(y^n||x^{n-1})=p(y^n)$, then $L_2(Y^n\to X^n)=L_1(Y^n\to
X^n)$.
\end{proof}

This result can be readily extended to the general stock market, for
which the penalty is {\em upper bounded} by $L_2(Y^n\to X^n)$.

\subsection{Compression with Joint Distribution Mismatch}

In Section \ref{s_data_compreession} we investigated the cost of
ignoring forward and backward links when compressing a jointly
$(X^n,Y^n)$ by an optimal lossless variable length code. Here we
investigate the penalty of assuming forward and backward links
incorrectly when neither exists. Let $X^n$ and $Y^n$ be independent
sequences. Suppose we compress them with a scheme that would have been optimal under the incorrect
assumption that the forward link $p(y^n||x^n)$ exists.  The optimal
lossless average variable length code under these assumptions
satisfies (up to 1 bit per source symbol)
\begin{align}
E(L(X^n,Y^n))&=\sum_{x^n,y^n}p(x^n)p(y^n)\log \frac{1}{p(y^n||x^n)p(x^n)}\nonumber \\
&= \sum_{x^n,y^n}p(x^n)p(y^n)\log \frac{p(x^n)p(y^n)}{p(y^n||x^n)p(x^n)}+H(X^n)+H(Y^n)\nonumber \\
&= \sum_{x^n,y^n}p(x^n)p(y^n)\log \frac{p(y^n)}{p(y^n||x^n)}+H(X^n)+H(Y^n)\nonumber \\
&=L_1(X^{n}\to Y^n)+H(X^n)+H(Y^n).
\end{align}
Hence the penalty is $L_1(X^{n}\to Y^n)$. Similarly, if we
incorrectly assume that the backward link $p(x^{n}||y^{n-1})$
exists, then
\begin{align}
E(L(X^n,Y^n))&=\sum_{x^n,y^n}p(x^n)p(y^n)\log \frac{1}{p(x^n||y^{n-1})p(y^n)}\nonumber \\
&= \sum_{x^n,y^n}p(x^n)p(y^n)\log \frac{p(x^n)p(y^n)}{p(x^n||y^{n-1})p(y^n)}+H(X^n)+H(Y^n)\nonumber \\
&= \sum_{x^n,y^n}p(x^n)p(y^n)\log \frac{p(x^n)}{p(x^n||y^{n-1})}+H(X^n)+H(Y^n)\nonumber \\
&= L_1(Y^{n-1}\to X^{n})+H(X^n)+H(Y^n).
\end{align}
Hence the penalty is $L_1(Y^{n-1}\to X^n)$. If both links are
mistakenly assumed, the
penalty\cite{Palomar_verdu08LautumInformation} is lautum information
$L(X^n; Y^n)$.  Note that the penalty due to wrongly assuming both
links is the sum of the penalty from wrongly assuming each link.
This is due to the conservation law (\ref{e_cons_lautum}).

\subsection{Hypothesis Testing}

We revisit the hypothesis testing problem in Section~V, which is describe in Fig. \ref{f_hypothesis}. As a dual to
Theorem 6, we characterize the minimum type I error exponent given
the type II error probability:

\begin{theorem}[Chernoff--Stein lemma for the causal dependence test: Type I error]
\label{t_Chernoff-Stel_lautum}
Let $(X,Y) = \{X_i,Y_i\}_{i=1}^\infty$ be stationary, ergodic, and
Markov of finite order such that $p(x^n,y^n) = 0$ implies
$p(x^n||y^{n-1})p(y^n) = 0$.  Let $A_n\subseteq {(\mathcal X\times
  \mathcal Y)}^n$ be an acceptance region, and let $\alpha_n$ and
$\beta_n$ be the corresponding probabilities of type I and type II
errors (\ref{e_error_def}).  For $0<\epsilon<\frac{1}{2}$, let
\begin{equation}\label{e_alpha_n}
\alpha_n^{(\epsilon)}=\min_{A_n\subseteq {(\mathcal X\times \mathcal Y)}^n, \beta_n<\epsilon}
\alpha_n.
\end{equation}
Then
\begin{equation}
\lim_{n\to\infty} -\frac{1}{n} \log \alpha_n^{(\epsilon)}=\mathcal
L_2(X\to Y),
\end{equation}
where the directed lautum information rate is the one induced by the joint
distribution from $H_0$, i.e., $p(x^n||y^{n-1})p(y^n||x^n)$.
\end{theorem}

The proof of Theorem \ref{t_Chernoff-Stel_lautum} follows very similar
steps as in Theorem \ref{t_Chernoff-Stel} upon letting
\begin{equation}
A_n^c=\left\{x^n,y^n:\left|\log \frac{p(y^n)}{p(y^n||x^n)}-\mathcal
L_2( X\to Y)\right|<\delta \right\},
\end{equation}
analogously to (\ref{e_A_n}), and upon using the Markov assumption for
guaranteeing the AEP; hence we omit the details.

\section{Concluding Remarks}
We have established the role of directed information in portfolio
theory, data compression, and hypothesis testing. Put together with
its key role in communications~\cite{Kramer98, Kramer03,
PermuterWeissmanGoldsmith09,
  TatikondaMitter_IT09, Kim07_feedback, PermuterWeissmanChenMAC_IT09,
  ShraderPemuter07ISIT,Pradhan07Venkataramanan}, and in estimation~\cite{PermuterKimTsachy_ContinousDirectedConCom09},
directed information is thus emerging as a key information theoretic
entity in scenarios where causality and the arrow of time are crucial
to the way a system operates. Among other things, these findings
suggest that the estimation of directed information can be an
effective diagnostic tool for inferring causal relationships and
related properties in a wide array of problems. This direction is
under current investigation~\cite{Lei09directed_estimation}.

\section*{Acknowledgment}
The authors would like to thank Ioannis Kontoyiannis  for
helpful discussions.

\bibliographystyle{unsrt}

\end{document}